\documentclass[runningheads]{llncs}
\usepackage{graphicx}
\usepackage{amsmath,amssymb,amsfonts}
\usepackage{mathrsfs}
\usepackage{diagbox}
\newcommand{\B}[1]{\boldsymbol #1}
\begin{document}
\title{Conditions for Unnecessary Logical \\ Constraints in Kernel Machines\thanks{This is a post-peer-review, pre-copyedit version of an article published in LNCS, volume 11728. The final authenticated version is available online at: \texttt{https://doi.org/10.1007/978-3-030-30484-3\_49}}}
%
%
\author{Francesco Giannini \and Marco Maggini}

\authorrunning{F. Giannini \and M. Maggini}
%
\institute{Department of Information Engineering and Mathematical Sciences\\
	University of Siena, Siena, via Roma 56, Italy \\
\email{\{fgiannini,maggini\}@diism.unisi.it} }
\maketitle              
\begin{abstract}
	A main property of support vector machines consists in the fact that only a small portion of the training data is significant to determine the maximum margin separating hyperplane in the feature space, the so called \emph{support vectors}. In a similar way, in the general scheme of learning from constraints, where possibly several constraints are considered, some of them may turn out to be unnecessary with respect to the learning optimization, even if they are active for a given optimal solution. In this paper we extend the definition of support vector to {\em support constraint} and we provide some criteria to determine which constraints can be removed from the learning problem still yielding the same optimal solutions. In particular, we discuss the case of logical constraints expressed by \L ukasiewicz logic, where both inferential and algebraic arguments can be considered. Some theoretical results that characterize the concept of {\em unnecessary constraint} are proved and explained by means of examples.

\end{abstract}

\section{Introduction}

Support vector machines (SVMs) are a class of kernel methods originally conceived by Vapnik and Chervonenkis \cite{Cortes1995}. One of the main advantages of this approach is the capacity to create nonlinear classifiers by applying the kernel trick to maximum--margin hyperplanes \cite{boser1992training,Cortes1995}. This property derives from the implicit definition of a (possibly infinite) high--dimensional feature representation of data determined by the chosen kernel. In the supervised case, the learning strategy consists in the optimization of an objective function, given by a regularization term, subject to a set of constraints that enforce the membership of the example points to the positive or negative class, as specified by the provided targets. The satisfaction of these constraints can be obtained also by the minimization of a hinge loss function that does not penalize output values ``beyond'' the target. As a consequence, the solution of the optimization problem will depend only on a subset of the given training data, namely those that contribute to the definition of the maximum--margin hyperplane separating the two classes in the feature space. In fact, if we approach the problem in the framework of constrained optimization, these points will correspond to the {\em active} constraints in the Lagrangian formulation. This means that we can split the training examples into two categories, the \emph{support vectors}, that completely determine the optimal solution of the problem, and the \emph{straw vectors}. By solving the Lagrangian dual of the optimization problem, the support vectors are those supervised examples corresponding to constraints whose Lagrangian multiplier is not null. In this paper we extend this paradigm to a class of semi--supervised learning problems where logical constraints are enforced on the available samples. 
\vspace{-0.05cm}

Learning from constraints has been proposed in the framework of kernel machines as an approach to combine prior knowledge and learning from examples \cite{gnecco2015foundations}. In particular, some techniques to exploit knowledge expressed in a description logic language \cite{cumby2003kernel} and by means of first-order logic (FOL) rules have been proposed in the literature \cite{muggleton2005support,diligenti2012bridging}.  In general, these techniques assume a multi--task learning paradigm where the functions to be learnt are  subject to a set of logical constraints, which provide an expressive and formally well--defined representation for abstract knowledge. For instance, logical formulas may be translated into continuous functions by means of t-norms theory \cite{diligenti2015semantic}. This mapping allows the definition of an optimization problem that integrates supervised examples and the enforcement of logical constraints on a set of available groundings. In general, the resulting optimization problem is not guaranteed to be convex as in the original SVM framework due to contribution of the constraints. However, it turns out to be convex when considering formulas expressed with a fragment of the {\L}ukasiewicz logic \cite{giannini2017learning}. In this case, the problem can be formulated as quadratic optimization since the constraints are convex piece-wise linear functions. Other related methods to embed logical rules into learning schemes have been considered, such as \cite{serafini2016learning,serafini2016logic}, where a framework called Logic Tensor Networks has been proposed, and \cite{hu2016harnessing}, where logic rules are combined with neural network learning.

The notion of {\em support constraints} has been proposed in \cite{melacci2011,gori2013constraint} to provide an extension of the concept of support vector when dealing with learning from constraints. The idea is based on the definition of entailment relations among constraints and the possibility of constraint checking on the data distribution. In this paper, we provide a formal definition of {\em unnecessary constraints} that refines the concept of support constraint and we provide some theoretical results characterizing the presence of such constraints. These results are illustrated by examples that show in practice how the conditions are verified. The main idea is that unnecessary constraints can be removed from a learning problem without modifying the set of optimal solutions. Similarly, with the specific goal to define algorithms accelerating the search for solutions in optimization problems, it is worth to mention the works in the Constraint Reduction (CR) field. In particular, in \cite{jung2012adaptive} it is shown how to reduce the computational burden in a convex optimization problem by considering at each iteration the subset of the constraints that contains only the most critical (or necessary) ones. In this sense,  our approach allows us to determine theoretically which are the unnecessary constraints as well as to enlighten their logical relations with the other constraints.

The paper is organized as follows. In Section \ref{sec:prob} we introduce the notation and the problem formulation. Then, Section \ref{sec:scm} analyzes the structure of the optimal solutions, providing the conditions to determine the presence of unnecessary constraints. The formal definition of {\em unnecessary constraint} and the related theorems are reported in Section \ref{sec:unnecess}. In Section \ref{sec:examples} we show how the proposed method is applied by means of some examples and finally, some conclusions and future directions are discussed in Section \ref{sec:conc}.

\section{Learning from Constraints in Kernel Machines}\label{sec:prob}

We consider a multi--task learning problem with ${\bf P}=\{p_j:\mathbb{R}^{n_j}\rightarrow\mathbb{R}:\,j\leq J\}$ denoting a set of $J>0$ functions to be learned. We assume that each $p_j$ belongs to a \emph{Reproducing Kernel Hilbert Space} (RKHS)~\cite{paulsen2016introduction} $\mathcal{H}_j$ and it is expressed as 
\[
p_j(x)=\omega_j\cdot\phi_j(x)+b_j \ ,
\]	
where $\phi_j$ is a function that maps the input space into a feature space (possibly having infinite dimensions), such that $k_j(x,y)=\phi_j(x)^T\cdot\phi_j(y)$, where $k_j\in\mathcal{H}_j$ is the $j$-th kernel function. The notation is quite general to take into account the fact that predicates (f.i. unary or binary) can be defined on different domains and approximated by different kernel functions\footnote{Predicates sharing the same domain may be approximated in the same RKHS by using the same kernel function.}. 

We assume a semi-supervised scheme in which each $p_j$ is trained on two datasets, $\mathscr{L}_j$ containing the supervised examples  and $\mathscr{U}_j$ containing the unsupervised ones, while all the available input samples for $p_j$ are collected in $\mathscr{S}_j$, as follows 
\[
\begin{array}{l}
\mathscr{L}_j=\{(x_l,y_l):\,l\leq l_j,x_l\in\mathbb{R}^{n_j},y_l\in\{-1,+1\}\}, \\
\mathscr{U}_j=\{x_u:\,u\leq u_j,x_u\in\mathbb{R}^{n_j}\}, \\
\mathscr{S}_j=\{x_s:\,s\leq s_j\}=\{x_l:\,(x_l,y_l)\in\mathscr{L}_j\}\cup\mathscr{U}_j,\quad S=\sum_{j=1}^Js_j \ .
\end{array}
\]
In the following, whenever we write $p_j(x_s)$, we assume $x_s\in\mathscr{S}_j$. Functions in {\bf P} are assumed to be predicates subject to some prior knowledge expressed by a set of \emph{First--Order Logic} (FOL) formulas $\varphi_h$ with $h\leq H$ in a knowledge base $KB$, and evaluated on the available samples for each predicate.

\subsection{Constraints}
The learning problem is formulated to require the satisfaction of three classes of constraints, defined as follows.
\begin{itemize}
	\item \emph{Consistency} constraints derive from the need to limit the values of predicates into $[0,1]$, in order to be consistent with the logical operators:
	\[
	0\leq p_j(x_s)\leq 1,\qquad x_s\in\mathscr{S}_j,\,j\leq J \ .
	\]
	\item \emph{Pointwise} constraints derive from the supervisions by requiring the output of the functions to be 1 for target $y_l=1$ and 0 for $y_l=-1$:
	\[
	y_l(2 p_j(x_l)-1)\geq 1,\qquad (x_l,y_l)\in\mathscr{L}_j,\,j\leq J \ .
	\]
	\item \emph{Logical} constraints are obtained by mapping each formula $\varphi_h$ in KB into a continuous real-valued function $f_h$ according to the operations of a certain t-norm fuzzy logic\footnote{See e.g. \cite{hajek1998metamathematics}  for more details on fuzzy logics.}  (see Tab. \ref{tab:tnorms} for the  \L ukasiewicz fuzzy logic) and then
	forcing their satisfaction by
	\[
	1-f_h(\B p)\leq0,\qquad h\leq H \ ,
	\]
	where for any $j\leq J$, $\B p_j=[p_j(x_1),\ldots,p_j(x_{s_j})]\in[0,1]^{s_j}$ is the vector of the evaluations (groundings) of the $j$-th predicate on the samples in $\mathscr{S}_j$ and $\B p=[\B p_1,\ldots \B p_J]\in[0,1]^S$ is the concatenation of the groundings of all the predicates.
\end{itemize}


\begin{table}[t]
	\centering
	\begin{tabular}{|c|c|c|c|c|c|}
		\hline
		$\neg x$ & $x\otimes y$ & $x\wedge y$  & $x\vee y$ & $x\oplus y$  & $x\Rightarrow y$ \\
		\hline
		  $1-x$ & $\max\{0,x+y-1\}$ & $\min\{x,y\}$ & $\max\{x,y\}$ & $\min\{1,x+y\}$ & $\min\{1,1-x+y\}$ \\
		\hline
	\end{tabular}
	\vspace{0.1cm}
	\caption{Logic connectives and their algebraic semantics for the \L ukasiewicz logic. From left to right: \emph{negation, strong conjunction (t-norm), weak conjunction, weak disjunction, strong disjunction (t-conorm), implication (residuum).}}
	\label{tab:tnorms}
	\vspace{-0.5cm}
\end{table}

\subsection{Optimization Problem}	

Given the previously defined constraints, the learning problem can be formulated as \emph{primal} optimization as, 
\begin{problem}\label{eq:pripro}
	\[
	\small
	\displaystyle\min_{\omega_j}\; \frac{1}{2}\sum_{j=1}^J||\omega_j||^2\qquad \mbox{subject to:}
	\]
	\vspace{-0.1cm}
	\[
	\small
	\begin{array}{lcl} 
	0\leq p_j(x_s)\leq 1,\qquad& &\qquad  \mbox{for }x_s\in\mathscr{S}_j,\;j\leq J \\
	y_l(2 p_j(x_l)-1) \geq 1,\qquad& &\qquad \mbox{for } (x_l,y_l)\in\mathscr{L}_j,\;j\leq J \\
	1-f_h(\B p)\leq 0,\qquad& &\qquad \mbox{for }h\leq H
	\end{array}
	\]
\end{problem}
This problem was shown to be solvable by quadratic optimization provided the formulas in $KB$ belong to the convex \L ukasiewicz fragment (i.e. formulas exploiting only the operators $(\wedge,\oplus)$ in Tab.~\ref{tab:tnorms} \cite{giannini2018convex}) and, in the following, we keep this assumption. This yields the functional constraints to be both convex and piecewise linear functions, hence they can be expressed as the max of a set of $I_h$ affine functions\footnote{The number of linear pieces $I_h$ depends on both the formula and the number of groundings used in that formula.} (see Theorem 2.49 in \cite{rockafellar2009variational})
\begin{equation}\label{eq:constr}
1-f_h(\B p)=\max_{i\leq {I_h}}(M_{h,i}\cdot \B p+q_{h,i})
\end{equation}
where $M_{h,i}=[m^{h,i}_{1,1},\ldots,m^{h,i}_{1,s_1},\ldots,m^{h,i}_{J,s_J}]\in\mathbb{R}^S$ is a vector defining the $i$-th linear piece depending on the structure of the $h$-th formula, and $q_{h,i}\in\mathbb{R}$. Basically any $m^{h,i}_{j,s}$ weighs the contribution of the $s$-th sample in $\mathscr{S}_j$ for the $j$-th predicate in the $i$-th linear piece deriving from the  \L ukasiewicz formula of the $h$-th logic constraint. The matrix $M$, obtained concatenating all the $M_{h,i}$ by row, may have several null elements, as shown in the examples reported in the following.
\begin{example}\label{ex:convtran}
Let $p_1, p_2$ be a unary and a binary predicate, respectively, evaluated on $\mathscr{S}_1=\{x_1,x_2\}$ and $\mathscr{S}_2=\{(x_1,x_1),(x_1,x_2),(x_2,x_1),(x_2,x_2)\}=\mathscr{S}_1 \times \mathscr{S}_1$ so that $\B p_1=[p_1(x_1),p_1(x_2)]$, $\B p_2 =[p_2(x_1,x_1),p_2(x_1,x_2),p_2(x_2,x_1),p_2(x_2,x_2)]$ denote their grounding vectors. Given the formula $\varphi=\forall x\,\forall y\, (p_1(x)\wedge p_1(y))\Rightarrow p_2(x,y)$, according to the convex \L ukasiewicz operators, its corresponding functional constraint $1-f_{\varphi}$ can be rewritten as the max of a set of affine functions, i.e. $\max_{x,y}\{0,\,p_1(x)+p_1(y)-p_2(x,y)-1\}$, that can be made explicit with respect to the grounding vectors of $p_1$ and $p_2$ by:
	\[
	\small
	\begin{array}{l}
	\max\{0,\,2p_1(x_1)-p_2(x_1,x_1)-1,\,p_1(x_1)+p_1(x_2)-p_2(x_1,x_2)-1,\\
	p_1(x_1)+p_1(x_2)-p_2(x_2,x_1)-1,\,2p_1(x_2)-p_2(x_2,x_2)-1\} \ .
	\end{array}
	\]
	In this case $I_h=5$,  and, for instance, $M_{\varphi,2}=[2,0,-1,0,0,0]$ and $q_{h,1}=-1$.
\end{example}
According to eq. (\ref{eq:constr}), any logical constraint $1-f_h(\B p)\leq0$ for $h\leq H$ can be replaced by $I_h$ linear constraints $M_{h,i}\cdot \B p+q_{h,i}\leq 0$, yielding \emph{Problem \ref{eq:pripro}} to be reformulated as quadratic programming. Hence, assuming to satisfy the associated KKT--conditions and that the feasible set of solutions is not empty, for any $j\leq J$ the optimal solution obtained by differentiating the Lagrangian function of \emph{Problem \ref{eq:pripro}} (see \cite{giannini2018convex}) is computed as:
\begin{eqnarray}\label{eq:sol}
\small
p_j^*(x)= 2\sum_{l=1}^{l_j}\lambda^*_{j_l}y_lk_j(x_l,x)-\sum_{h=1}^H\sum_{i=1}^{I_h}\lambda^*_{h_i}\cdot\sum_{s=1}^{s_j}m^{h,i}_{j,s}\cdot k_j(x_s,x)+\nonumber\\
+\sum_{s=1}^{s_j}(\eta^*_{j_s}-\bar{\eta}^*_{j_s})k_j(x_s,x)+b^*_j  \ .
\end{eqnarray}
Each solution can be written as an expansion of the $j$-th kernel $k_j$ with respect to the three different types of constraints on the corresponding sample points. As in classical SVMs, we may study the constraints whose optimal Lagrange multipliers $\lambda^*_{j,l}, \lambda^*_{h,i}, \eta^*_{j,s}, \bar{\eta}^*_{j,s}$ are not null, namely the {\em support (active) constraints}.



\section{Unnecessary Constraints}\label{sec:scm}


The optimal solution of \emph{Problem \ref{eq:pripro}} is determined only by the support constraints. The problem is convex if the Gram matrix of the chosen kernel is positive-semidefinite and strictly convex if it is positive-definite. The solution is guaranteed to be unique only in this second case \cite{boyd2004convex}. For both cases, different multiplier vectors $\lambda$, $\eta$ may yield an optimal solution for the Lagrangian function associated to the problem, e.g. see \emph{Example \ref{ex:trans2}}.

 In this study, we are interested in constraints that are \emph{not necessary} for the optimization, even if they may turn out to be active for a certain solution. The main results of this paper establish some criteria to discover unnecessary constraints and their relationship with the underling consequence relation among formulas in \L ukasiewicz logic.

\subsection{About Multipliers for Logical Constraints}\label{sec:sscm}
By construction, pointwise and consistency constraints are both related to a single sample $x$ for a  given predicate. This means that the contribution of the active constraints of this type in any point is weighted by a specific multiplier, as expressed by the first and third summations in eq. $(\ref{eq:sol})$. On the opposite, each logical constraint involves in general more predicates eventually evaluated on different points (each Lagrange multiplier in the second summation of eq. $(\ref{eq:sol})$ is associated to a set of samples). Hence we may wonder if it exists a vector of Lagrange multipliers yielding the same contribution to the solution for each point, for which as much as possible multipliers are null. 

For simplicity, eq.  $(\ref{eq:sol})$ can be rewritten more compactly by grouping the terms with respect to any sample $x_s$ as
\begin{equation}
\small
p^*_j(x)=\sum_{s=1}^{s_j}\left(\alpha^{*(P)}_{j,s}+\alpha^{*(L)}_{j,s}+\alpha^{*(C)}_{j,s}\right)k_j(x_s,x)=\sum_{s=1}^{s_j}\alpha^*_{j,s}k_j(x_s,x)
\end{equation}
where $\alpha^{*(P)}_j(\lambda^*_{j,l}), \alpha^{*(L)}_j(\lambda^*_{h,i}), \alpha^{*(C)}_j(\eta^*_{j,s},\bar{\eta}^*_{j,s})$ denote the vectors of optimal coefficients (depending on optimal Lagrange multipliers) of the kernel expansion for pointwise, logical and consistency constraints respectively.
In particular, the term for the logical constraints is defined as
\begin{equation}\label{eq:logcoef}
\small
\alpha^{*(L)}_{j,s}=\sum_{h=1}^H\sum_{i=1}^{I_h}\lambda^*_{h,i}m^{h,i}_{j,s} \ .
\end{equation}
 Since (\ref{eq:logcoef}) corresponds to the overall contribution of the logical constraints to the $j$-th optimal solution in its $s$-th point, we are interested in the case where we obtain the same term with different values for the multipliers $\lambda^*_{h,i}$. In particular, we would like to verify if there exists $\bar{h}\leq H$ such that for every $j\leq J$ and for every $s\leq s_j$, it is possible to compute $\bar{\lambda}_{h,i}$ such that
\begin{equation}\label{eq:pr2}
\small
\alpha^{*(L)}_{j,s}=\sum_{\substack{h=1 \\ h\neq\bar{h}}}^H\sum_{i=1}^{I_h}\bar{\lambda}_{h,i}m^{h,i}_{j,s} \ .
\end{equation}
This condition yields the same solution to the original problem but without any direct dependence on the $\bar{h}$-th constraint.  This case can be determined as defined in the following \emph{Problem \ref{eq:prob}}, where a matrix formulation is considered, and then by looking for a solution (if there exist) with null components for the $\bar{h}$-th constraint. 

\begin{problem}\label{eq:prob}
	Given an optimal solution $\alpha^*$ for \emph{Problem \ref{eq:pripro}}, find $\lambda\in\mathbb{R}^N$ such that
	\[
	M\cdot \lambda= \alpha^* ,
	\]
	where $N=\sum_{h=1}^H I_h$ and $M=[M_{1,1},\ldots,M_{h,I_h}]\in\mathbb{R}^{S\times N}$.
\end{problem}
Let $v_1,\ldots,v_n$ be  an orthonormal base  of the space generated by $Ker(M)=\{\lambda:\,M\cdot\lambda=0\}$, such that any solution can be expressed as 
\[
\lambda=\lambda^*+\sum_{i=1}^nt_iv_i \ ,
\]
for some $t_i\in\mathbb{R}$. We have the following cases:
\begin{description}
	\item[(i)] if $dim(Ker(M))=0$ then the system allows the unique solution $\lambda^*$;
	\item[(ii)] if $dim(Ker(M))\neq 0$ then there exist infinite solutions.
\end{description}
In the first case, the only constraints whose multipliers give null contribution to the optimal solution are the original straw constraints. Whereas in the second case, we look for a solution $\bar{\lambda}$ (if there exists) where $\bar{\lambda}_{\bar{h},i}=0$ for any $i\leq {I_{\bar{h}}}$ for some $\bar{h}\leq H$. Indeed in such a case, we can replace $\lambda^*$ with $\bar{\lambda}$ by transferring the contribution of the $\bar{h}$-th constraint to the other constraints still obtaining the same optimal solution for the predicates. This is carried out by solving the linear system with $I_{\bar{h}}$ equations $\bar{\lambda}_{\bar{h},i}=0$ and $n$ variables $t_1,\ldots,t_n$. 

%
\begin{remark}
	In the following, we will say that a vector $(\lambda_{h,i})_{h\leq H,\,i\leq I_h}$ is a solution of \emph{Problem \ref{eq:prob}} with respect to $\bar{h}$, if it is a solution and $\lambda_{\bar{h},i}=0$ for every $i\leq I_{\bar{h}}$.
\end{remark}

\subsection{Unnecessary Hard--Constraints}\label{sec:unnecess}

Roughly speaking, we say that a given constraint is \emph{unnecessary} for a certain optimization problem if its enforcement does not affect the solution of the problem. The main idea is that if we consider two problems (defined on the same sample sets and with the same loss), one with and one without the considered constraint, both have the same optimal solutions. 
The relation between logical inference and deducible constraints arises naturally in this frame, indeed logical deductive systems involve truth-preserving inference. 
In addition, logical constraints are quite general to include both pointwise and consistency constraints. A supervision $(x_l,y_l)$ for a predicate $p$ can be expressed by $1\rightarrow p(x_l)$ if $y_l=1$ and by $p(x_l)\rightarrow 0$ if $y_l=-1$, while the consistency constraints by $(0\rightarrow p(x_l))\wedge(p(x_l)\rightarrow 1)$. We note that in this uniform view, \emph{Problem \ref{eq:prob}} applies to all the constraints.

%

\begin{definition}\label{def:unnecessary}
	Let us consider the learnable functions in {\bf P} evaluated on a sample $\mathscr{S}$ and $KB=\{\varphi_1,\ldots,\varphi_H\}$.  We say that $\varphi_{\bar{h}}\in KB$ is \emph{unnecessary} for ${\bf HP}$ if the optimal solutions of problems ${\bf HP}$ and ${\bf \overline{HP}}$ coincide, where
	\[
	\small
	({\bf HP}) \; \min_{\alpha} Loss(\alpha), \mbox{ with } 1-f_h(\B p)\leq0,\mbox{ for }h\leq H \ ,
	\]
	\[
	\small
	({\bf \overline{HP}}) \; \min_{\alpha} Loss(\alpha), \mbox{ with } 1-f_h(\B p)\leq0,\mbox{ for }h\leq H, h\neq\bar{h}
	\]
	$Loss(\alpha)=\displaystyle\sum_{j\leq J}\alpha'_jK_j\alpha_j$ and $K_j=\left(k_j(x_i,x_k)\right)_{i,k\leq s_j}$ is the Gram matrix of $k_j$. 
	
	.
\end{definition}
If $\mathcal{F}$ and $\overline{\mathcal{F}}$ are the feasible sets of ${\bf HP}$ and ${\bf \overline{HP}}$ respectively, we have $\mathcal{F}\subseteq\overline{\mathcal{F}}$, however in general they are not the same set.
%

Since all the considered constraints correspond to logical formulas, we can also exploit some consequence relation among formulas in \L ukasiewicz logic. In the following, we will write $\Gamma\models\phi$, where $\Gamma\cup\{\phi\}$ is a set of propositional formulas, to express the true-preserving logical consequence in {\bf \L}, stating that $\phi$ has to be evaluated as true for any assignation satisfying all the formulas in $\Gamma$.
\begin{proposition}\label{teo:log}
	If $\{\varphi_h:\,h\leq H,\,h\neq\bar{h}\}\models\varphi_{\bar{h}}$ then $\varphi_{\bar{h}}$ is unnecessary for ${\bf HP}$.
\end{proposition}
\begin{proof}
	By hypothesis, any solution satisfying the constraints of ${\bf \overline{HP}}$ satisfies the constraints of ${\bf HP}$ as well, namely we have $\mathcal{F}=\overline{\mathcal{F}}$. The conclusion easily follows since the two problems have the same loss function with the same feasible set.
\end{proof}
	One advantage of this approach is providing some criteria to determine the constraints that are not necessary for a learning problem. Indeed, in presence of a large amount of logical rules, \emph{Proposition \ref{teo:log}} guarantees we can remove all the deducible constraints simplifying the optimization still getting the same solutions.

The vice versa of \emph{Proposition \ref{teo:log}} is not achievable, since the logical consequence has to hold for every assignation. The notion of unnecessary constraint is local to a given dataset, indeed the available sample is limited and fixed in general. However, if a constraint is unnecessary then the optimal solutions with or without it coincide and we have that such constraint is satisfied whenever the other ones are satisfied by any optimal assignations. Such consequence among constraints, taking into account only the assignations leading to best solutions on a given dataset, provides an equivalence with the notion of unnecessary constraint. It is interesting to notice that a slightly different version of this consequence has already been  considered in \cite{melacci2011}.

\subsection{Towards an Algebraic Characterization}
In Sec. \ref{sec:sscm} we introduced a criterion to discover if a given constraint $\varphi_{\bar{h}}$ can be deactivated solving \emph{Problem \ref{eq:prob}}. The method consists in finding a vector of Lagrange multipliers with null components corresponding to $\varphi_{\bar{h}}$. We are now interested in discovering the relation between this criterion and the notion of unnecessary constraint. Some results are stated by the following propositions.

\begin{proposition}
	If $\varphi_{\bar{h}}$ is unnecessary for ${\bf HP}$ then for any optimal solution of this problem there exists a $KKT$-solution $\bar{\lambda}$ of \emph{Problem \ref{eq:prob}} with respect to $\bar{h}$. 
\end{proposition}
\begin{proof}
	If $\varphi_{\bar{h}}$ is unnecessary then ${\bf HP}$ and ${\bf \overline{HP}}$ have the same optimal solutions. Let us consider one of them, lets say $\alpha^*$, where  $\alpha^*=\alpha(\lambda^*_{h,i})=\alpha(\hat{\lambda}^*_{h,i})$ for the two problems with respect to some multipliers vectors $(\lambda^*_{h,i})_{h\leq H,\,i\leq I_h}$ and $(\hat{\lambda}^*_{h,i})_{h\leq H,\,h\neq\hat{h},\,i\leq I_h}$. 
	Since the two vectors of multipliers yield the same optimal solution, then we can define for every $h\leq H,i\leq I_{h}$ a solution still satisfying the KKT-conditions (also called a KKT-solution) $\bar{\lambda}$ of \emph{Problem \ref{eq:prob}} as:
	\[ 
	\small
	\bar{\lambda}_{h,i}=
	\begin{cases} 
	\hat{\lambda}^*_{h,i} & \mbox{for }h\neq\bar{h} \\
	0 & \mbox{otherwise} \ .
	\end{cases}
	\]
\end{proof}

This has to be thought of as a necessary condition to discover which logical constraints can be removed from ${\bf HP}$ still preserving its optimal solutions. However, the other way round does hold in case either ${\bf HP}$ or ${\bf \overline{HP}}$ has a unique solution $\alpha^*$, but in general we can only prove a weaker result.

\begin{proposition}\label{teo:inc}
	If there exists a $KKT$-solution $\bar{\lambda}$ of \emph{Problem \ref{eq:prob}} with respect to $\bar{h}$ (for a certain optimal solution $\bar{\alpha}^*=\alpha(\bar{\lambda})$ of ${\bf HP}$), then the set of optimal solutions of ${\bf HP}$ is included in the set of optimal solutions of ${\bf \overline{HP}}$. 		
\end{proposition}
\begin{proof}
	Given any optimal solution $\alpha^*$ of ${\bf HP}$, since the problem is (at least) convex, we have $Loss(\alpha^*)=Loss(\bar{\alpha}^*)$. At this point, we note that $\bar{\alpha}^*$ is also feasible for ${\bf \overline{HP}}$ and that the restriction of $\bar{\lambda}$ on components $h\neq\bar{h}$ is a vector of Lagrange multipliers for ${\bf \overline{HP}}$ satisfying the KKT-conditions. The convexity of the problem guarantees that the KKT-conditions are sufficient as well. This means that $\bar{\alpha}^*$ is also an optimal solution for ${\bf \overline{HP}}$, hence its loss value is a global minimum and the same holds for $\alpha^*$.
\end{proof}
In this case we can not conclude that any optimal solution of ${\bf \overline{HP}}$ is an optimal solution for {\bf HP} because in general this solution could be not feasible for this problem. However as we pointed out above, we have the following result.
\begin{corollary}\label{teo:coro}
	If either ${\bf HP}$ or equivalently ${\bf \overline{HP}}$ has a unique solution then the premise of {\bf Proposition \ref{teo:inc}} is also sufficient.		
\end{corollary}
\begin{proof}
	The solution is unique if the Gram matrix K, that is the same in both the problems, is positive-definite. Hence, requiring the uniqueness of the solution for the two problems is equivalent and the claim is trivial from {\bf Proposition \ref{teo:inc}}.
\end{proof}


\section{Some Examples}\label{sec:examples}
Here we illustrateq
, by means of some cases solved in MATLAB with the interior-point-convex algorithm, how the method works and we discuss the results to clarify what described so far. In particular, we exploit the \emph{transitive law} as an example to enlighten how the presented theoretical results apply.

\begin{example}\label{ex:trans}
	We are given the predicates $p_1,p_2,p_3$ subject to $\forall x\,p_1(x)\rightarrow p_2(x)$, $\forall x\, p_2(x)\rightarrow p_3(x)$, $\forall x\,p_1(x)\rightarrow p_3(x)$.
	Given a common evaluation dataset $\mathscr{S}$, the logical formulas can be translated into the following linear constraints
	\[
	\small
	\begin{array}{ccc}
	\max_{x\in\mathscr{S}}\{0,p_1(x)-p_2(x)\},\;&
	\max_{x\in\mathscr{S}}\{0,p_2(x)-p_3(x)\},\;&
	\max_{x\in\mathscr{S}}\{0,p_1(x)-p_3(x)\}
	\end{array}
	\]
	and yield the following terms for the Lagrangian associated to \emph{Problem  \ref{eq:pripro}},
	\[
	\lambda_{1,1}(p_1(x_1)-p_2(x_1)),\ldots,\lambda_{3,s}(p_1(x_s)-p_3(x_s)) \ .
	\]
	At first we solve the optimization problem where, to avoid trivial solutions, we provide few supervisions for the predicates and we exploit a polynomial kernel. To keep things clear, we consider only two points defined in $\mathbb{R}^2$, $\mathscr{S}=\{(1,0.5),\,(0.4,0.3)\}$. Hence, given the solution $\alpha(\lambda^*)$ (uniqueness holds) of \emph{ Problem \ref{eq:pripro}} (see Fig. \ref{fig:2points}), where $\lambda^*=(0.5549,0,0,0.5706,0,0)$, we have
	
	\[
	\small
	M_2=\left(
	\begin{array}{cccccc}
	1 & 0 & 0 & 0 & 1 & 0 \\
	0 & 1 & 0 & 0 & 0 & 1 \\
	-1 & 0 & 1 & 0 & 0 & 0 \\
	0 & -1 & 0 & 1 & 0 & 0 \\
	0 & 0 & -1 & 0 & -1 & 0 \\
	0 & 0 & 0 & -1 & 0 & -1 
	\end{array}
	\right) ,
	\mbox{ }  
	\alpha=M_2\cdot\lambda^*=
	\left(
	\begin{array}{c}
	\lambda^*_{1,1}+ \lambda^*_{3,1} \\
	\lambda^*_{1,2}+ \lambda^*_{3,2} \\
	-\lambda^*_{1,1}+ \lambda^*_{2,1} \\
	-\lambda^*_{1,2}+ \lambda^*_{2,2} \\
	-\lambda^*_{2,1}- \lambda^*_{3,1} \\
	-\lambda^*_{2,2}- \lambda^*_{3,2} 
	\end{array}
	\right)=
	\left(
	\begin{array}{c}
	0.5549 \\
	0 \\
	-0.5549 \\
	0.5706 \\
	0 \\
	-0.5706 
	\end{array}
	\right) \ .
	\]
	In this case all the solutions of \emph{ Problem \ref{eq:prob}} are given for any $t_1,t_2\in\mathbb{R}$ by
	\[
	\small
	\lambda=\lambda^*+t_1\cdot\left(
	\begin{array}{c}
	-1 \\
	0 \\
	-1 \\
	0 \\
	1 \\
	0 
	\end{array}
	\right)+t_2\cdot\left(
	\begin{array}{c}
	0 \\
	-1 \\
	0 \\
	-1 \\
	0 \\
	1 
	\end{array}
	\right)=
	\left(
	\begin{array}{c}
	\lambda^*_{1,1}-t_1 \\
	\lambda^*_{1,2}-t_2 \\
	\lambda^*_{2,1}-t_1 \\
	\lambda^*_{2,2}-t_2 \\
	\lambda^*_{3,1}+t_1 \\
	\lambda^*_{3,2}+t_2 \\
	\end{array}
	\right)=\left(
	\begin{array}{c}
	0.5549-t_1 \\
	-t_2 \\
	-t_1 \\
	0.5706 -t_2\\
	t_1 \\
	t_2 
	\end{array}
	\right)
	\ ,
	\]
	where the pair of vectors $v_1=(-1,0,-1,0,1,0)',\,v_2=(0,-1,0,-1,0,1)'$ is a base for $Ker(M_2)$. From this, we get that the only way to obtain the same $\alpha$ nullifying the contribution of the third constraint is taking $t_1=t_2=0$, namely taking $\lambda=\lambda^*$. It is worth to notice that we can also decide to nullify the contribution of the first or of the second constraint taking $t_1=0.5549,t_2=0$ or $t_1=0,t_2=0.5706$. In these cases we get $\lambda^*_1=(0,0,-0.5549,0.5706,0.5549,0)'$, $\lambda^*_2=(0.5549,-0.5706,0,0,0,0.5706)'$, but the third one is a support constraint.
\end{example}

\begin{figure}[t]
	\centering
	\begin{minipage}{5cm}
		\includegraphics[width=5cm]{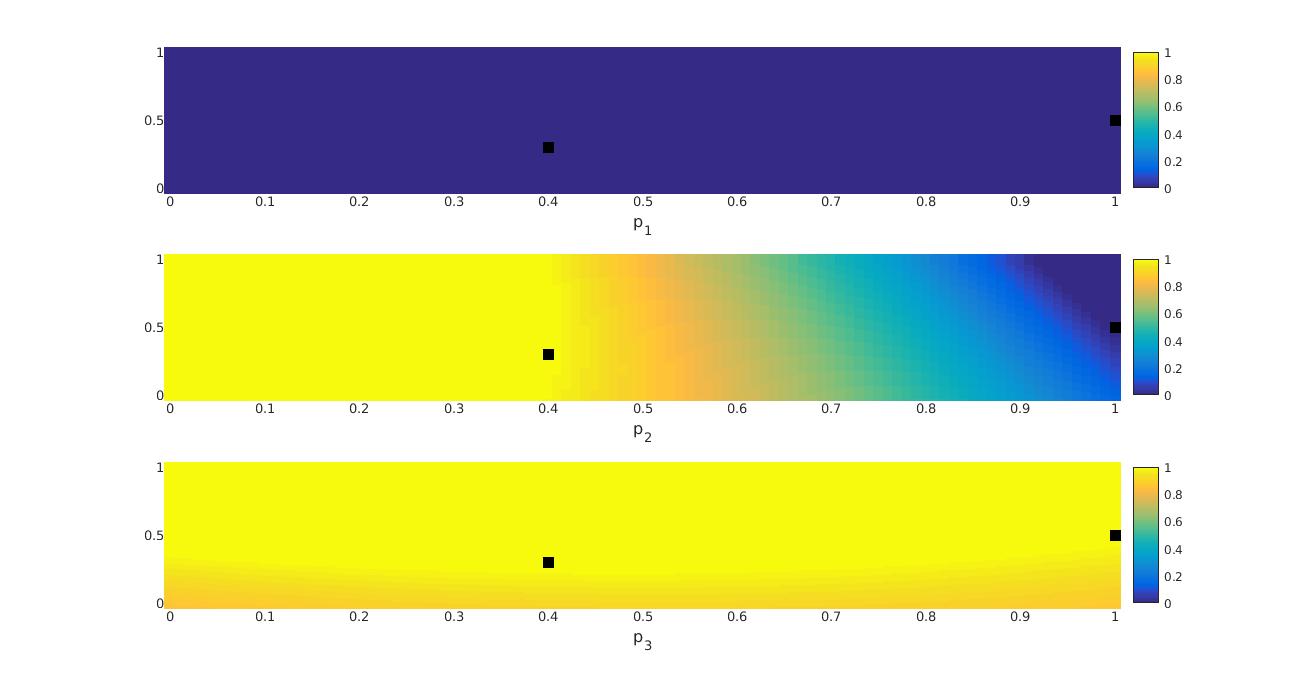}
	\end{minipage}
	\hspace{0.1cm}
	\begin{minipage}{5cm}
		\includegraphics[width=5cm]{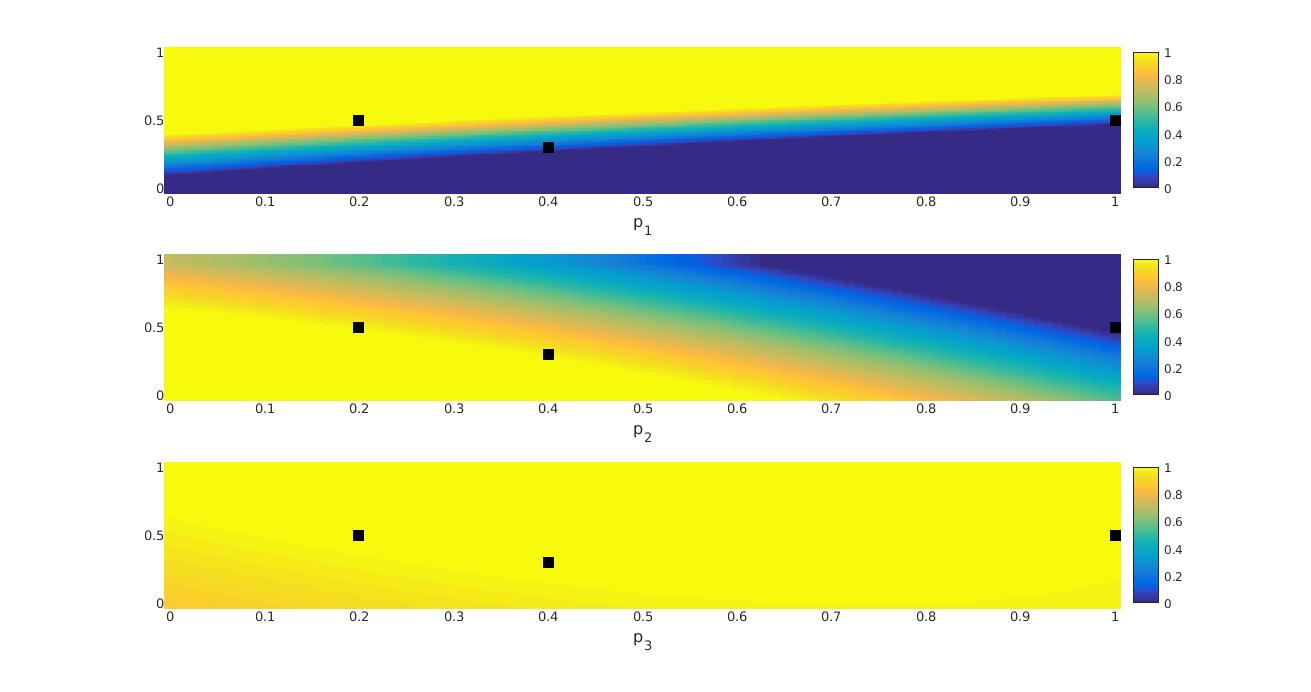}
	\end{minipage}
	\caption{From left to right we report the evaluation of the learnt functions $p_1,p_2,p_3$ in the example space for \emph{Example \ref{ex:trans}} and \emph{Example \ref{ex:trans2}}, respectively. Filled squares correspond to the provided sample points.}
	\label{fig:2points}
\end{figure}



 Although it is easy to see that the third constraint is deducible from the other ones, \emph{ Problem \ref{eq:pripro}} may give a different perspective in terms of support constraints.

\begin{example}\label{ex:trans2}
	Given the same problem as \emph{Example \ref{ex:trans}} with the additional point $(0.2,0.5)$ in $\mathscr{S}$, we get $\lambda^*=(0.3520, 0.3453, 0, 1.1529, 0, 0.5631, 0.4202, 0, 0)$, hence the third constraint turns out to be initially supporting. However we may wonder if there is another solution of \emph{ Problem \ref{eq:prob}} where the components of the third constraint are null. 
	The matrix $M_3$ is obtained from $M_2$ by adding three rows and three columns corresponding to the additional grounding of the predicates and to the components for the logical constraints on the new point. 
	\[
	\small
	M_3=
	\left(
	\begin{array}{ccccccccc}
	1 & 0 & 0 & 0 & 0 & 0 & 1 & 0 & 0 \\
	0 & 1 & 0 & 0 & 0 & 0 & 0 & 1 & 0 \\
	0 & 0 & 1 & 0 & 0 & 0 & 0 & 0 & 1 \\
	-1 & 0 & 0 & 1 & 0 & 0 & 0 & 0 & 0 \\
	0 & -1 & 0 & 0 & 1 & 0 & 0 & 0 & 0 \\
	0 & 0 & -1 & 0 & 0 & 1 & 0 & 0 & 0 \\
	0 & 0 & 0 & -1 & 0 & 0 & -1 & 0 & 0 \\
	0 & 0 & 0 & 0 & -1 & 0 & 0 & -1 & 0 \\
	0 & 0 & 0 & 0 & 0 & -1 & 0 & 0 & -1
	\end{array}
	\right)
	\]
	In this case, the dimension of $Ker(M_3)$ is increased exactly by one, as the number of affine components of any involved logical constraint. This means, we can try to find a $\bar{\lambda}^*$ in which a certain constraint has null values. For instance, the vector $\bar{\lambda}^*$= $(0.7722, 0.3453, 0, 1.5731,$ $ 0, 0.5631, 0, 0, 0)$ is a solution of \emph{Problem \ref{eq:prob}} with respect to the third constraint. However, as in \emph{Example \ref{ex:trans}}, it is the only KKT-solution allowing us to remove the contribution of a constraint.
\end{example}


\subsection{From Support to Necessary Constraints}
 Combining pointwise and consistency constraints brings any optimal solution to be evaluated exactly to 0 or 1 on any supervised sample and all the corresponding Lagrange multipliers to be different from zero, namely they will turn out to be support constraints. However, they could be unnecessary constraints for the problem and we could actually remove them from the optimization. 

\begin{example}
	We consider the same problem as \emph{Example \ref{ex:trans}} where  $\mathscr{S}=\{(0.4,0.3)\}$ is labelled as negative for $p_1$ and positive for both $p_2$ and $p_3$.  We express the pointwise and the consistency constraints in logical form. All the constraints are obtained requiring the following linear functions to be less or equal to zero:
	\[
	\small
	\begin{array}{ccc}
	(\mbox{logical})&(\mbox{pointwise})&(\mbox{consistency})\\
	&&\\
	p_1(x_1)-p_2(x_1), & p_1(x_1), & -p_1(x_1), p_1(x_1)-1,  \\
	p_2(x_1)-p_3(x_1), & 1-p_2(x_1), & -p_2(x_1), p_2(x_1)-1, \\
	p_1(x_1)-p_3(x_1), & 1-p_3(x_1), & -p_3(x_1), p_3(x_1)-1.
	\end{array}
	\]
	Exploiting the complementary slackness and the condition for the Lagrange multipliers given by \emph{Problem \ref{eq:prob}}, we can provide several combinations of values for the multipliers yielding the same solution. The Gram matrix $K$ is positive-definite ($K=1.25$) and the solution $\alpha^*=(0,0.8,0.8)$ provided by a linear kernel is unique. For this simple example we have only two possible KKT-solutions of \emph{Problem \ref{eq:prob}} minimizing the number of necessary constraints, they are $\bar{\lambda}=(0,0,0,0,0,0,0,0,0,0.8,0,0.8)'$ and $\hat{\lambda}=(0,0.8,0,0,0,0,0,0,0,0,0,1.6)'$. 
	This may be easily shown since the complementarity slackness force $\lambda_1=\lambda_3=\lambda_8=\lambda_9=\lambda_{11}=0$ and multiplying by $M$ the remaining multipliers, they have to satisfy:
	\[
	\small
	\begin{cases}
	\lambda_4-\lambda_7=0 \\
	\lambda_2-\lambda_5+\lambda_{10}=0.8 \\
	-\lambda_2-\lambda_6+\lambda_{12}=0.8  \ .
	\end{cases}
	\]
	Since {\bf HP} has a unique solution, from {\bf Corollary \ref{teo:coro}}, we have two different minimal optimization problems. One with only $p_2(x_1)-1\leq0$ and $p_3(x_1)-1\leq0$ as necessary constraints and the other with only $p_2(x_1)-p_3(x_1)\leq0$ and $p_3(x_1)-1\leq0$ once again.
\end{example}

\section{Conclusions}\label{sec:conc}

In general, in learning from constraints, several constraints are combined into an optimization scheme and often it is quite difficult to identify the contribution of each of them. In particular, some constraints could turn out to be not necessary for finding a solution. In this paper, we propose a formal definition of unnecessary constraint as well as a method to determine which are the unnecessary constraints for a learning process in a multi-task problem. The necessity of a certain constraint is related to the notion of consequences among the other constraints that are enforced at the same time. This is a reason why we suppose to deal with logical constraints that are quite general to include both pointwise and consistency constraints. The logical consequence among formulas is a sufficient condition to conclude that a constraint, corresponding to a certain formula, is unnecessary. However, we also provide an algebraic necessary condition that turns out to be sufficient in case the Gram matrices associated to the kernel functions are positive-definite. 


%

\bibliographystyle{splncs04}
\bibliography{references}

\end{document}